\documentclass[conference]{IEEEtran}
\IEEEoverridecommandlockouts
\usepackage{amsmath,amssymb,amsfonts,amsthm}
\usepackage{algorithm}
\usepackage{algorithmic}
\usepackage{graphicx}
\usepackage{textcomp}
\usepackage{xcolor}
\usepackage{cite}
\usepackage{verbatim}
\usepackage{makecell}
\usepackage{booktabs}
\usepackage{arydshln}
\newcommand{\ra}[1]{\renewcommand{\arraystretch}{#1}}

\usepackage{graphicx}
\usepackage{float} 
\usepackage{stfloats}
\usepackage{subfigure}

\usepackage{enumitem}

\usepackage{caption}

\newtheorem{theorem}{Theorem}[section]

\newtheorem{definition}{Definition}[section]

\usepackage{threeparttable}    

\usepackage{tcolorbox}
\usepackage{pifont}
\usepackage{setspace}

\def\BibTeX{{\rm B\kern-.05em{\sc i\kern-.025em b}\kern-.08em
    T\kern-.1667em\lower.7ex\hbox{E}\kern-.125emX}}
\begin{document}

\title{Secure Outsourced Decryption for FHE-based Privacy-preserving Cloud Computing\\
}

    \author{
    \IEEEauthorblockN{Xirong Ma$^1$, Chuan Li$^1$, Yuchang Hu$^4$, Yunting Tao$^{1,3}$, Yali Jiang$^1$,\\ Yanbin Li$^1$, Fanyu Kong$^{1,2,}$\IEEEauthorrefmark{1}\thanks{\IEEEauthorrefmark{1} Corresponding author: Fanyu Kong, Email: fanyukong@sdu.edu.cn}, Chunpeng Ge$^{1,5}$} \\
    
    \IEEEauthorblockA{$ ^1$ School of Software, Shandong University, Jinan, 250101, Shandong, China}
    \IEEEauthorblockA{$ ^2$ Shandong Sansec Information Technology Co., Ltd, Jinan, 250101, Shandong, China}
    \IEEEauthorblockA{$ ^3$ College of Information Engineering, Binzhou Polytechnic, Binzhou, 256603, Shandong, China}
    \IEEEauthorblockA{$ ^4$ College of Artificial Intelligence, Nanjing Agricultural University, Nanjing, 210095, Jiangsu, China}
    \IEEEauthorblockA{$ ^5$ Joint SDU-NTU Centre for Artificial Intelligence Research (C-FAIR), Jinan, 250101, Shandong, China}
    }



\maketitle

\begin{abstract}

The demand for processing vast volumes of data has surged dramatically due to the advancement of machine learning technology. Large-scale data processing necessitates substantial computational resources, prompting individuals and enterprises to turn to cloud services. Accompanying this trend is a growing concern regarding data leakage and misuse. Homomorphic encryption (HE) is one solution for safeguarding data privacy, enabling encrypted data to be processed securely in the cloud. However, the encryption and decryption routines of some HE schemes require considerable computational resources, presenting non-trivial work for clients. In this paper, we propose an outsourced decryption protocol for the prevailing RLWE-based fully homomorphic encryption schemes. The protocol splits the original decryption into two routines, with the computationally intensive part executed remotely by the cloud. Its security relies on an invariant of the NTRU-search problem with a newly designed blinding key distribution. Cryptographic analyses are conducted to configure protocol parameters across varying security levels. Our experiments demonstrate that the proposed protocol achieves up to a $67\%$ acceleration in the client's local decryption, accompanied by a $50\%$ reduction in space usage.

\end{abstract}


\section{Introduction} \label{sec: Intro}
In recent years, machine learning has witnessed profound advancements across various domains, with the models growing in size and complexity. Individuals and organizations gradually turn to cloud services for more efficient and convenient model training and data storage.
However, the untrusted identity of the cloud raises concerns among clients regarding their data privacy. Users are apprehensive that their uploaded data might be leaked or misused during its processing on the cloud. 

Homomorphic encryption is a promising technique capable of addressing these concerns and safeguarding data privacy. Fully homomorphic encryption (FHE), in particular, supports arbitrary computation on ciphertexts. Its development has seen rapid strides over the past 15 years\cite{gentry2009fully}. Presently, the prevailing FHE schemes \cite{cheon2017homomorphic,chillotti2020tfhe,brakerski2014leveled,fan2012somewhat} can facilitate secure computation of complex machine learning operations, such as neural network inference and training \cite{sav2020poseidon,froelicher2020scalable,lou2020glyph,lee2022privacy}. These schemes are built based on the RLWE (Ring-Learning with Error) security assumption \cite{lyubashevsky2010ideal,lyubashevsky2013toolkit}, which relates to difficult lattice-based problems and may provide potential resilience against the attacks of quantum algorithms.


A classic example of secure cloud computing utilizing homomorphic encryption is illustrated in Figure \ref{fig: scenario1}. In this scenario, the client transmits its homomorphically encrypted data to the cloud. The cloud then processes the encrypted data according to the client's instruction $f$ and returns the encrypted result to the client, which can be decrypted using the secret key $sk_A$. Despite the delegation of task $f$ to the cloud server, clients have to frequently perform encryption and decryption operations to upload data and recover results. For RLWE-based HE schemes mentioned above, these routines involve polynomial multiplications with large degrees and modulus, posing a nontrivial challenge for resource-constrained devices. 


\begin{figure}
\centerline{\includegraphics[scale=0.42]{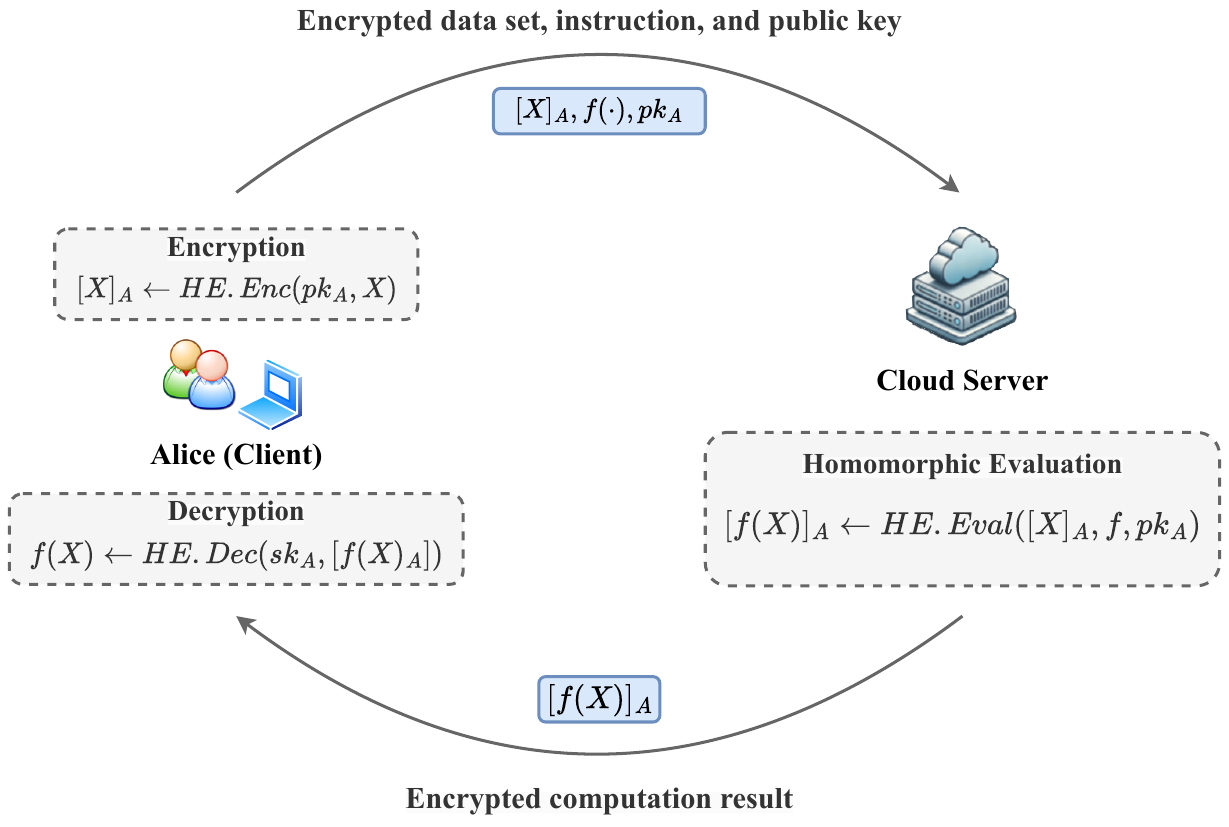}} 
\caption{Secure cloud computing scenario using homomorphic encryption}
\label{fig: scenario1}
\end{figure}


To ease clients' computational burden, we propose a new protocol for outsourcing the decryption routine. It divides the original decryption algorithm into two stages that can be performed collaboratively by both the cloud server and the client. The stage executed by the cloud is referred to as 'blind decryption', while the stage executed by the client is called 'local decryption'. The latter is more efficient and requires less memory than the original decryption process. Our contributions are outlined as follows:

\begin{itemize}[leftmargin = 1em]
    \item We devise a multiplication-communicative secret key blinding technique as the core of the outsourced decryption protocol. This allows us to securely delegate part of the decryption computation to the cloud using a blinded secret key. The proposed blinding method utilizes the inverses of sparse polynomials as blinding factors, enabling unblinding to be accomplished more swiftly than executing dense polynomial multiplication, which constitutes the primary complexity of the decryption algorithm.

    \item We analyze the security of the outsourced decryption protocol. The proposed protocol is proven to be semantically secure under RLWE and NTRU assumptions, wherein the blinded secret key can be interpreted as an instance of the NTRU-search problem. We further perform various attacks to parameterize the security of our protocol, given that the blinding factors are drawn from a new distribution distinct from conventional NTRU settings\cite{hoffstein2017choosing,hoffstein1998ntru}.
    
    \item We implement the protocol in the renowned homomorphic encryption library Lattigo\cite{lattigo} and compare its performance with the original decryption algorithm, which shows the efficiency of our protocol. 

\end{itemize}

\subsection{Roadmap}
We recall the related work in Section \ref{sec: rlwork}. Section \ref{sec: archi} gives an overview of the outsourcing decryption protocol. Section \ref{sec: OtsDec-D} presents the detailed design of the protocol. We discuss the security in Section \ref{sec: secanl}. Section \ref{sec:exp} presents experiment results. Finally, we summarize this research in Section \ref{sec: conclusion}.



\section{Background and Related Work}\label{sec: rlwork}\label{sec: prelimary}\label{sec: GLWE}
\subsection{Homomorphic Encryption}

One of the most prevailing and versatile types of (fully) homomorphic encryption \cite{cheon2017homomorphic,fan2012somewhat,brakerski2014leveled,chillotti2020tfhe} is built based on the RLWE security assumption \cite{regev2009lattices,lyubashevsky2010ideal,lyubashevsky2013toolkit}. It includes HE schemes that support homomorphic operations in various plaintext spaces, including complex vector space\cite{cheon2017homomorphic,cheon2019full}, modulo integer vector space\cite{gentry2012homomorphic,brakerski2014leveled,fan2012somewhat}, $\{0, 1\}$ binary space\cite{ducas2015fhew,chillotti2020tfhe}, and others. The performance of these RLWE-based HE schemes is determined by the efficiency of polynomial ring element operations, among which polynomial multiplication may be the most important factor. The Number Theoretic Transformation(NTT) is now widely deployed in HE schemes for polynomial multiplication acceleration\cite{cheon2019full,halevi2019improved}, boasting an asymptotic complexity of $O(d\log{d})$, where $d$ signifies the polynomial degree. Nonetheless, polynomial multiplication remains a non-trivial computation in HE schemes since $d$ may range from $2^{10}$ to $2^{16}$ along with the length of the coefficient modulus spanning from tens to thousands \cite{bossuat2024security}. Typically, encryptions and decryptions in RLWE-based HE schemes require polynomial multiplications, so it is worth thinking about whether there is any way to reduce their computational burden for clients in some cloud computing scenarios. Our work may provide some interesting insight on this topic.

In the following, we briefly introduce the RLWE-based public key encryption (PKE) \cite{regev2009lattices,brakerski2012fully,fan2012somewhat,cheon2017homomorphic}, which is the foundation of RLWE-based HE schemes, where routines supporting homomorphic computation are straightly built upon this PKE system (but we do not go this far since we only focus on decryption).

Given a set of pre-determined parameters: $(d,q,\chi_{sec},\chi_{err})$, the RLWE-based HE cryptosystem is defined on a polynomial ring $R_q = Z_q[X]/(X^d+1)$. To encrypt a plaintext $m \in R_q$, the encryption algorithm goes like this:
\begin{equation}
    (u,v) \in R_q^2 \leftarrow (ar+e_2,(as+e_1)r + \Delta_1 m + \Delta_2 e_3).
\end{equation}
Here,$s$ is the secret key sampled from distribution $\chi_{sec}$ over $R_q$, and $(a,b) = (a, as+e) \in R_q^2$ is the public key. $a$ is uniformly sampled from $R_q$, and $e_1,e_2,e_3\in \chi_{err}$ are some noisy polynomials inserted to ensure security. $\Delta_1$ and $\Delta_2$ are some scaling factors to control the precision of the plaintext and the noise growth in the ciphertext. The core procedure in decryption is an inner product over $R_q^2$:
\begin{equation}\label{eq: dec}
    \Delta_1 m + \Delta_2 e \leftarrow \left< (u,v), (s,1) \right> \in R_q 
\end{equation}
By properly setting $\Delta_1$ and $\Delta_2$, one may correctly recover $m$ from the decryption result using simple scalar multiplication and rounding techniques. For example, the renowned BFV (Brakerski\/ Fan and Vercauteren) scheme\cite{brakerski2012fully,fan2012somewhat} sets $\Delta_1 m= \lfloor \frac{q}{t} m \rfloor$ and $\Delta_2 = 1$, and its entire decryption can be written as $m \leftarrow \lceil \frac{t}{q} \left< (u,v), (s,1) \right> \rfloor \mod{t}$, where $m$ is in the plaintext space $R_t$.



\subsection{Sparse Polynomial Operation}

Sparse polynomials possess a distinct representation from dense polynomials. The former is represented by a series of pairs comprising non-zero coefficients and their corresponding degrees, while the latter utilizes fixed-length arrays. In practice, the disparity in the volume of stored information ensures that sparse polynomials inherently have a more compact representation and, to some extent, greater computational efficiency. A recent survey\cite{Roche2018what} reports that existing algorithms\cite{van2012complexity,roche2008adaptive,roche2011chunky,van2013bit,arnold2013faster,Arnold2016SparsePI} can achieve a complexity of $O(k\log{d})$ for sparse polynomial multiplication under certain conditions, where $k$ is the maximum number of non-zero terms in the input and output polynomials and $d$ the upper bound of the degree. 

However such complexity may increase to $O(d\log{d})$ when the number of non-zero terms in the two input polynomials diverges significantly, i.e., when one input has $k_1 < \log{d}$ terms and the other has $k_2=d$ terms. We term this as the "sparse-to-dense" multiplication. In this scenario, the most straightforward implementation achieves a lower complexity of $O(k_1d)$. 




Moreover, sparse polynomials are also of interest in cryptographic scheme design. Fixed-sparsity binary (or ternary) polynomials were initially used as keys in the NTRU cryptosystem\cite{hoffstein1998ntru}, where their Hamming weight was employed to gauge the success probability of decryption and the security of the scheme. In RLWE-based homomorphic encryption schemes, coefficients of private key polynomials are typically sampled from ternary distributions or discrete Gaussian distributions. Distributions sparser than the ternary distribution can support a larger homomorphic computation budget, yet their security implications have not been thoroughly explored\cite{bossuat2024security,albrecht2021standard}.

\subsection{Secure Outsourcing of Cryptographic Computations}
Certain cryptographic algorithms exhibit significant complexity, rendering lightweight devices unsuitable for executing such resource-intensive algorithms. One solution is to securely outsource these algorithms to high-performance machines. Hohenberger \textit{et al.}\cite{Hohenberger} proposed a methodology and framework for this approach. Subsequently, much research has provided various solutions for outsourcing cryptography-related computations, such as modular exponentiation \cite{ma2013ModExp,chen2013ModExp,di2017ModExp,Su2017InvMod,di2019ModExp,ye2016ModExp,hu2020ModExp,Su2020ModExp,rath2024ModExp}, modular inversion \cite{Su2017InvMod,tian2019ModInv,Bouillaguet2023ModInv}, bilinear pairings\cite{CHEN2015bilinear,Tian2015bilinear,zhang2021bilinear,ying2024bilinear}, and scalar multiplication on elliptic curves\cite{zhou2016elliptic,Yuan2020elliptic}. 


Nevertheless, to our knowledge, there has been a relatively limited focus on outsourcing RLWE-based cryptographic computations involving large-scale polynomial multiplication. From the perspective of polynomial operations, the most relevant research we can find is the outsourcing scheme of the extended Euclidean algorithm proposed by Zhou \textit{et al.}\cite{ZHOU2020XGCD}, which can be seen as an extension of outsourcing the modular inversion. However, the computations they outsourced are not at the same granularity as those discussed in this paper, where they measure complexity in terms of polynomial multiplication, while we consider the complexity of multiplication itself. Moreover, their method to blind the input polynomials differs from ours, as they employ a variable substitution method for blinding.

\section{Protocol Overview \& Security Model} \label{sec: archi}

The proposed outsourced decryption protocol consists of a tuple of routines, they are $\Pi_{\text{Setup}}$, $\Pi_{\text{SkBdKeyGen}}$, $\Pi_{\text{SkBd}}$, $\Pi_{\text{BldDec}}$, and $\Pi_{\text{LocDec}}$:\\
\begin{itemize}[leftmargin = 1em]
    \item $\Pi_{\text{Setup}}$\textbf{(Setup)}: Takes parameters $(d,q,\chi_{sk},\lambda)$ from the underlying HE scheme as input and outputs a hamming weight $h$, where $d$ and $q$ denote the polynomial ring $R_q = \mathbb{Z}[X]_q/(X^d+1)$, $\chi_{sk}$ denotes the secret key distribution, and $\lambda$ denotes the security parameter.
    
    
    \item $\Pi_{\text{SkBdKeyGen}}$\textbf{(Blinding key generation for secret key)}: Takes a hamming weight $h$ as input and outputs a pair of factors $t^{-1},t \in R_q$ as the key pair for blinding and unblinding the secret key.
    
    \item $\Pi_{\text{SkBd}}$\textbf{(Secret Key Blinding)}: Takes a secret key $sk$ from the underlying HE scheme and the blinding key pair $(t^{-1},t)$ as input, and outputs the blinded secret key $\tilde{s}$.
    
    \item $\Pi_{\text{BldDec}}$\textbf{(Blind Decryption)}: Takes a ciphertext $c = (u,v) \in R_q^{2}$ and the blinded secret key $\tilde{s}$ as input, and outputs the blind decryption of $c$.
    
    \item $\Pi_{\text{LocDec}}$\textbf{(Local Decryption)}: Takes a blindly decrypted ciphertext $\tilde{c}$, perform unblinding on it and outputs a plaintext $m$.\\
\end{itemize}

\begin{figure}
    \centering
    \tcbset{width=\columnwidth, colframe=black,colback=white,arc=0mm}

\begin{tcolorbox}
\centerline{\textbf{Setup Phase}}

\textbf{Input}:\quad The client $\mathcal{T}$ inputs the ring $R_q = Z_q[X]/(X^d+1)$, the secret key $s\leftarrow \chi_{sk}$, and the security parameter $\lambda$ inherited from the RLWE-based HE scheme established by itself; The cloud $U_{H}$ inputs $\perp$. 

\textbf{Output}:\quad $\mathcal{T}$ receives a pair of blinding keys for $s$: $(t,t^{-1})$; $\mathcal{U}_H$ receives a blinded version of $s$: $\tilde{s}$.\\

\begin{enumerate}[leftmargin = 1em]
    \item \textbf{[Protocol setup]}: $\mathcal{T}$ performs $\prod_{\text{Setup}}(d,q,\lambda,\chi_{sk})$ and outputs parameter $h$.
    \item \textbf{[Generate blinding key]}: $\mathcal{T}$ performs $\prod_{\text{SkBdKeyGen}}(d,q,h)$ and outputs $t^{-1},t \in R_q$. 
    \item \textbf{[Blinding]}: $\mathcal{T}$ performs $\prod_{\text{SkBd}}(s,t^{-1})$, outputs a blinded secret key $\tilde{s}$, and sends it to $\mathcal{U}_{H}$.\\
\end{enumerate}

\centerline{\textbf{Outsourced Decryption Phase}}

\textbf{Input}:\quad $\mathcal{T}$ inputs $\perp$; $\mathcal{U}_H$ inputs a set of ciphertexts $\{ ct_i \}$ encrypted under $s$, which is typically the result of some homomorphic evaluation and to be sent to $\mathcal{T}$ for decryption.

\textbf{Output}:\quad $\mathcal{T}$ receives the decryption result of $\{ct_i\}$; $\mathcal{U}_H$ receives $\perp$.\\

\begin{enumerate}[leftmargin = 1em]
    \item \textbf{[Blind Decryption]}: $\mathcal{U}_H$ initializes an empty set $P$. For each ciphertext $ct=(u,v)\in R_q^2$ in $\{ct_i \}$, $\mathcal{U}_H$ performs $\prod_{\text{BldDec}}(\tilde{s},ct)$ and adds the output into $P$. $\mathcal{U}_H$ then sents $P$ to $\mathcal{T}$.
    \item \textbf{[Local Decryption]}: On receiving blindly decrypted results $P$, for each $\tilde{pt}\in P$, $\mathcal{T}$ performs $\prod_{\text{LocDec}}(t,\tilde{pt})$ and accepts the output as the decryption result.\\
\end{enumerate}
\end{tcolorbox}
    \caption{Full description of the outsourced decryption protocol}
    \label{fig: scenario2}
\end{figure}

There are two entities involved in the outsourced decryption protocol:
\begin{itemize}[leftmargin = 1em]
    \item \textbf{Terminal/Client} $\mathcal{T}$: A resource-constraint terminal incapable of performing computationally heavy operations. $\mathcal{T}$ owns a pair of HE keys $(pk,sk)$ and seeks help for intensive computation from an external server.
    \item \textbf{Cloud server for homomorphic evaluation} $\mathcal{U}_{H}$: Server $\mathcal{U}_{H}$ handles computationally heavy homomorphic computation. Usually, this server may hold the public key $pk$ sent by $\mathcal{T}$ along with some encrypted data of $\mathcal{T}$, receive computation instruction from $\mathcal{T}$, and return the corresponding result to $\mathcal{T}$ for decryption.
\end{itemize}

As shown in Figure \ref{fig: scenario2}, the proposed protocol has two \textit{phases}, Setup and Outsourced decryption. During the Setup phase, the client $\mathcal{T}$ blinds his or her secret key and sends it to the cloud server. Subsequently, whenever $\mathcal{U}_{H}$ finishes a computation task ordered by $\mathcal{T}$, it performs blind decryption on the encrypted results and sends it back to $\mathcal{T}$ for local decryption. 

Similar to many privacy-preserving computation schemes using HE\cite{jiang2018secure,sav2020poseidon,MA2024103658}, we employ the semi-honest threat model for the proposed protocol, wherein the server $\mathcal{U}_{H}$ is assumed to be a passive adversary that adheres to the protocol but still maintains curiosity about $\mathcal{T}$'s private data. 

Note that verification of the outsourced decryption is required if we are to consider scenarios with higher security demands. However, we will not delve into this in the current paper, as verifying the correctness of the homomorphic computation itself is a non-trivial task and remains a significant research challenge today\cite{fiore2020boosting,bois2021flexible,ganesh2023rinocchio,chatel2024verifiable}. Nevertheless, we will discuss feasible approaches to verify the correctness of the proposed protocol at the end of the paper for future research reference.

For our outsourcing decryption protocol to be applicable to the aforementioned scenario, it must adhere to the following properties\cite{Hohenberger}: \\

\noindent \textbf{Semantic Security} We require that for any adversary $\mathcal{U}_{H}'$ having access to the public key $pk$ and the blinded secret key $\tilde{s}$, for any two messages $m_0,m_1 \in \mathcal{M}$, the advantage of $\mathcal{U}_{H}'$ in distinguishing between distributions $\text{HE.Enc}(pk,m_0)$ and $\text{HE.Enc}(pk,m_1)$ should be smaller than $2^{-\lambda}$. $\mathcal{M}$ denotes the plaintext space of the underlying HE scheme, and $\lambda$ is the security parameter.\\

\noindent \textbf{Correctness} For all arithmetic functions $f: \mathcal{M}^I \rightarrow \mathcal{M}$ defined in the underlying HE scheme, there exists an instance of the outsourcing decryption protocol such that $\Pi_{\text{LocDec}} (\Pi_{\text{BldDec}}( \text{HE.Eval}(f,pk,ct_1,\dots,ct_I), \tilde{sk} ), sk_{ob}) = f(m_1, \dots, m_I)$ holds with overwhelming probability.\\

\noindent \textbf{$\alpha$-Efficient Decryption} 
Considering $\mathcal{T}$ and $\mathcal{U}_H$ as algorithms composed of the routines they execute, the runtime of $\mathcal{T}(x)$ is no greater than an $\alpha$-multiplicative factor of the runtime of $\text{HE.Dec}(x)$ for any tuple of input $x$.\\ 

\section{Detailed Design of the Outsourced Decryption}\label{sec: OtsDec-D}

In this section, we introduce the detailed design of the routines in the outsourcing decryption protocol.

\subsection{Secret Key Blinding from Sparse Polynomial Inversion}

The complexity of the decryption algorithm (as described in Equation \ref{eq: dec}) is dominated by polynomial multiplication involving secret key and ciphertext components. Therefore, to securely outsource this computation to cloud computing, our idea is to devise a method for secret key blinding that maintains the commutative property of multiplication. This ensures that after the cloud computes the target multiplication using the blinded key (in place of the original key), the client can directly remove the blinding factor from the multiplication result to obtain the desired outcome. Specifically, for a given secret key $s$ we sample a pair of factor $(t,t^{-1})$ in $R_q$ and use them to perform blinding operation: $\tilde{s}\leftarrow st^{-1}$. It is evident that $\tilde{s} u t = su$ holds for any legal ciphertext $(u,v)\in R_q^2$. 

From the description above, one can observe that routines $\Pi_{\text{SkBd}}$ and $\Pi_{\text{BldDec}}$ are just normal polynomial multiplications (see Algorithms \ref{alg: seckeyob} and \ref{alg: pardec}). Note that the former one is invoked only once in the protocol, whose computational complexity can be amortized across multiple subsequent ciphertext decryption operations, making it negligible in the overall computational load. The latter is typically performed by $\mathcal{U}_{H}$ for each ciphertext to be decrypted.

The remaining challenge is to generate blinding factors that ensure $\Pi_{\text{LocDec}}$ operates with maximal efficiency. Note that the sparse-to-dense polynomial multiplication exhibits lower complexity compared to NTT-based dense multiplication when the sparsity and the circuit design are appropriately handled. Therefore, we choose sparse polynomials as the unblinding factors in $\Pi_{\text{SkBdKeyGen}}$. The sampling procedure is outlined in Algorithm \ref{alg: sample}, where we use hamming weight (denoted as $h$) to control the sparsity of a polynomial, and $S_h$ the set of all sparse polynomials over $R_q$ with hamming weight $h$. The unblinding factor is sampled from $S_h$ uniformly and randomly, with its inverse being the blinding factor. Note that an invertibility test should be done for the candidate unblinding factor since $R_q$ is merely a ring rather than a field.

$\Pi_{\text{SkBdKeyGen}}$ utilizes the Chinese Remain Theorem (CRT) to sample unblinding factors from $S_h$ conveniently. The following equation holds when $q$ can be decomposed into $L$ coprime factors $q = \prod_{i=0}^{L-1} q_i$:
\begin{equation}
    R_q \equiv R_{q_0} \times R_{q_1} \times \dots \times R_{q_{L-1}}.
\end{equation}
Elements individually satisfying the invertibility condition from different $R_{q_i}$ can be combined into an invertible element in $R_q$ using the CRT map (see Theorem \ref{thm: CRT1}). It can also be proven that elements sampled this way are uniformly distributed across $S_h$ (see Theorem \ref{thm: CRT2}). Using CRT can significantly reduce the complexity of the invertibility check since $q$ may be very large as we mentioned previously. Note that in homomorphic encryption schemes, the structure $q=\prod_{i=0}^L q_i$ is typically used to construct a large $q$, where $q_i$ are primes with similar length and satisfy the condition $q_i \equiv 1 \mod{2d}$ to enable NTT-based polynomial multiplication.




\begin{algorithm} 
    \renewcommand{\algorithmicrequire}{\textbf{Input:}}
    \renewcommand{\algorithmicensure}{\textbf{Output:}}
     \caption{Blinding key generation ($\Pi_{\text{SkBdKeyGen}}$)} \label{alg: sample}
    \begin{algorithmic}[1]
    \REQUIRE Ring $R_q = \mathbb{Z}_q[X]/X^d+1$; Hamming weight $h$ 
    \ENSURE Unblinding key $t$ with hamming weight $h$; Blinding key $t^{-1} \in R_q$
    \STATE Let $q = \prod_{i=0}^{L-1} q_i$ where $q_i$ is coprime to each other.
    \STATE $\mathbf{t} \leftarrow_U Z_d^h$
    \STATE $t \leftarrow 0 \in R_q$
    \STATE $t^{-1} \leftarrow 0 \in R_q$
    \FOR{$0\leq i <L$}
        \STATE $t^{(i)}\leftarrow 0 \in R_{q_i}$
        \STATE $p_i = q/q_i$
        \FOR{$ 0\leq j < h$} \label{step: resample}
            \STATE $t^{(i)}_j \leftarrow_U Z_{q_i}$
        \ENDFOR
        \STATE Check if $(t^{(i)})^{-1} \in R_{q_i}$ exists. If not, go to Step \ref{step: resample}. 
    \STATE $t \leftarrow t + t^{(i)}\cdot p_i \cdot (p_i^{-1} \mod{q_i}) \mod{q}$
    \STATE $t^{-1} \leftarrow t^{-1} + (t^{(i)})^{-1}\cdot p_i \cdot (p_i^{-1} \mod{q_i}) \mod{q}$
    \ENDFOR
    \RETURN $t,t^{-1} \in R_q$
    \end{algorithmic}
\end{algorithm}

\begin{theorem}\label{thm: CRT1}
    For any pair of $t,t^{-1} \in R_q$ sampled from Algorithm \ref{alg: sample} with sparsity $h$, $t\cdot t^{-1} \equiv 1 \mod{q}$
\end{theorem} 
\begin{proof}
    In Algorithm \ref{alg: sample}, we have that for $0\leq i <L$, $t^{(i)} \cdot (t^{(i)})^{-1} \equiv 1 \mod{q_i}$. According CRT, we have $t \cdot t^{-1} \equiv  t^{(i)}\cdot (t^{(i)})^{-1} \equiv 1 \mod{q_i}$. Then $t\cdot t^{-1}$ is equivalent to $1$ modulo the least common multiple of $q_0,\dots, q_L$, i.e. $q$. 
    
\end{proof}

\begin{theorem}\label{thm: CRT2}
    For any pair of $t,t^{-1} \in R_q$ sampled from Algorithm \ref{alg: sample} with sparsity $h$, $t$ is indistinguishable from any polynomial uniformly sampled from $R_q$ with sparsity $h$. 
\end{theorem} 
\begin{proof}
    Given that $q_i,i=0,\dots,L$ is coprime to each other, we can construct a vector $\mathbf{p} = [P_0,P_1,\dots,P_L]$, where $\{P_i = p_i\cdot (p_i^{-1}\mod{q_i})|i=0,\dots,L \}$. $\mathbf{p}$ forms a basis of $Z_q$. Then for any polynomial $a$ in $S_h$, it can be equivalently written by $\mathbf{a} = [a_0,a_1,\dots,a_L]$ where $\left< \mathbf{a}, \mathbf{p} \right> = a$. As Algorithm \ref{alg: sample} uniformly samples $\mathbf{a}$, it suffices to generate a random polynomial in $S_h$. 
\end{proof}

\subsection{Local Decryption}

$\Pi_{\text{LocDec}}$ employs the sparse-to-dense polynomial multiplication for the unblinding procedure, which leverages the sparsity of the unblinding factor to accelerate the multiplication (see Algorithm \ref{alg: LocDec}). The sparse-to-dense multiplication can be decomposed into $h\times L$ monomial multiplications since every component of the unblinding factor $t$ in different $R_{q_i}$ contains non-zero coefficients in the same positions. Monomial multiplications in $R_q$ can be easily conducted by coefficient shifting and scalar multiplication.

The performance of $\Pi_{\text{LocDec}}$ is further enhanced through several methods. Firstly, we delay the timing of the modulus operation when incorporating the result of monomial multiplications into $m$. This is due to our observation that when the loose upper bound of the infinite norm of $m^{(i)}$, i.e., $q_i \cdot h$, is significantly smaller than a word (e.g., 64 bits), we can safely postpone the modulus operation until the end of the aggregation of $m^{(i)}$ without the risk of overflow. Secondly, our implementation in the Lattigo library implies that minimizing assignment operations markedly improves performance.


\begin{algorithm} 
    \renewcommand{\algorithmicrequire}{\textbf{Input:}}
    \renewcommand{\algorithmicensure}{\textbf{Output:}}
    \caption{Secret Key Blinding ($\Pi_{\text{SkBd}}$)} \label{alg: seckeyob}
    \begin{algorithmic}[1]
    \REQUIRE Blinding key $s\in R_q$; Secret key $s \in R_q$
    \ENSURE Blinded secret key $\tilde{s} \leftarrow ts \in R_q$ 
    \end{algorithmic}
\end{algorithm}

\begin{algorithm} 
    \renewcommand{\algorithmicrequire}{\textbf{Input:}}
    \renewcommand{\algorithmicensure}{\textbf{Output:}}
    \caption{Blind Decryption ($\Pi_{\text{BldDec}}$)} \label{alg: pardec}
    \begin{algorithmic}[1]
    \REQUIRE Blinded secret key $\tilde{s}$; Ciphertext $c = (u,v) \in R_q^2$
    \ENSURE Blinded decryption result $\tilde{c} \leftarrow (u\tilde{s},v) \in R_q^2$
    \end{algorithmic}
\end{algorithm}

\begin{algorithm} 
    \renewcommand{\algorithmicrequire}{\textbf{Input:}}
    \renewcommand{\algorithmicensure}{\textbf{Output:}}
    \caption{Local Decryption ($\Pi_{\text{LocDec}}$)} \label{alg: LocDec}
    \begin{algorithmic}[1]
    \REQUIRE Unblinding key $t \in R_q$ with hamming weight $h$; Blinded decryption result $\tilde{c} = (\tilde{u},v)\in R_q$;
    \ENSURE Unblinded decryption result $m \in R_q$;
    \STATE Let $q = \prod_{i=0}^{L-1} q_i$ where $q_i$ is coprime to each other.
    \STATE $t^{(i)} \leftarrow t \mod{q_i}$
    \STATE Let vector $\mathbf{t}_{idx}$ represent the indices of $t$'s non-zero coeffecients.  
    \STATE For each $t^{(i)}$, let it be represneted a vector $\mathbf{t}^{(i)}$ where the $j$-th component corresponds to the non-zero coefficient at the $\mathbf{t}_{idx}[j]$-th term of $t^{(i)}$.
    \STATE $m\leftarrow 0 \in R_q$

    \FOR{$0\leq l < h$}
        \STATE $\text{monodegree} \leftarrow \mathbf{t}_{idx}[l]$
        \STATE $k_1\leftarrow d-\text{monodegree}$
        \STATE $k_2 \leftarrow -\text{monodegree}$
        \FOR {$q_i \in \{q_i|i=0,\dots,L\}$}
            \FOR{$0\leq j< \text{monodegree}$}
                \STATE $m[j] \leftarrow m[j] + (-\tilde{u}[k_1+j] \cdot \mathbf{t}^{(i)}[l]\mod{q_i})$
            \ENDFOR
            \FOR{$\text{monodegree} \leq j < d$}
                \STATE $m[j] \leftarrow m[j] + (\tilde{u}[k_2+j] \cdot \mathbf{t}^{(i)}[l]\mod{q_i})$
            \ENDFOR
        \ENDFOR
    \ENDFOR
    \FOR{$q_i \in \{q_i|i=0,\dots,L\}$}
        \FOR{$0 \leq j < d$ }
            \STATE $m[j]\leftarrow m[j] \mod{q_i}$
        \ENDFOR
    \ENDFOR
    \STATE $m\leftarrow m+v$
    \RETURN $m$
    \end{algorithmic}
\end{algorithm}

\subsection{Complexity Analysis}

$\mathcal{T}$ is responsible for routine $\Pi_{\text{SkBdKeyGen}},\Pi_{\text{SkBd}}$ and $\Pi_{\text{LocDec}}$ in the protocol, where the former two are only performed once following the generation of the secret key, and the later is done for each ciphertext blindly decrypted by $\Pi_{\text{BldDec}}$. Therefore, the complexity of $\mathcal{T}$ is dominated by the sparse-to-dense polynomial multiplication with roughly $hd$ (modular) scalar multiplications. 

On the other hand, the original decryption computed by $\mathcal{T}$ has a typical complexity of $ad\log{d} + d$ scalar multiplications corresponding to one Inversed NTT operation and coordinate-wise multiplication, where $a$ is a constant varies from concrete implementations of NTT. Then the efficiency of our proposed protocol is computed as $\frac{h}{a\log{d}+1}$.



\section{Security Analysis}\label{sec: secanl}

We initially present formal proof of the security of the proposed outsourced decryption in Theorem \ref{thm: OtsDecSec1}, reducing the protocol's security to the RLWE and NTRU assumptions. This stems from the fact that our blinded secret key can be viewed as a $(\gamma_s,\gamma_t,q)$-NTRU instance given in Definition \ref{def: NTRUinst}. Recovering the secret key from the blinded one corresponds to solving the $(\gamma_s,\gamma_t,q)$-search NTRU problem in Definition \ref{def: NTRUsearch}.

The major difference between our outsourcing decryption protocol and a standard NTRU public key encryption (PKE)\cite{hoffstein1998ntru,silverman2003Estimated,hoffstein2017choosing} is that an NTRU PKE typically requires $\gamma_s \approx \gamma_t \geq 1$ for correctness, whereas $\gamma_t$ in our protocol can be set as small as needed. Indeed, we want $t$ to span a sufficiently large space to ensure security and to be sparse enough to maintain efficiency. 

Recent research\cite{pellet2021hardness} indicates that classic hard lattice problems can be reduced to the NTRU-search problem, providing an upper bound on the difficulty of the NTRU-search problem. However, to parameterize the protocol's security, we have to analyze attacks on the secret key blinding technique, as the key distribution we employ differs from that targeted by previously known attacks \cite{coppersmith1997lattice,may1999cryptanalysis,silverman4dimension,albrecht2016subfield,cheon2016algorithm,duong2017choosing}. Below, we conduct several attacks on our secret key blinding technique, which we consider to be the most typical and threatening, and define the protocol's security based on the complexity of these attacks. A more comprehensive analysis of attacks on the protocol will be necessary in the future.



\begin{theorem}\label{thm: OtsDecSec1}
    The outsourced decryption is semantically secure, provided that the RLWE and NTRU-search assumptions hold where their instances are public together under the same secret.
\end{theorem}

\begin{proof}
    let $\mathcal{U}_{H}$ denotes adversary in the real world and $\mathcal{U}_{H}'$ the adversary in the ideal world. $\mathcal{U}_{H}$ takes in the public key $pk$, the blinded secret key $\tilde{s}$, a ciphertext $c$ and the auxiliary information represented by a polynomial $h$ as input, and guesses some knowledge of $c$ denoted as $f(c)$. $\mathcal{U}_{H}'$ also attempts to guess $f(c)$ but it has no access to $c$. To prove $\mathcal{U}_{H}'$ can correctly guess $f(c)$ with almost the same probability as $\mathcal{U}_{H}$, the simulation works as follows \cite{lindell2017simulate}: \\
    \begin{enumerate}[leftmargin=1em]
        \item $\mathcal{U}_{H}'$ runs:
            \begin{itemize}
                \item $pk,s \leftarrow \text{HE.KeyGen}(1^\lambda)$, 
                \item $(t,t^{-1})\leftarrow \Pi_{\text{SkBdKeyGen}}(1^\lambda)$,
                \item $\tilde{s}\leftarrow \Pi_{\text{SkBd}}(s,t^{-1})$.
            \end{itemize}
        \item $\mathcal{U}_{H}'$ computes $c \leftarrow \text{HE.Enc}(pk,m)$ where $m$ is sampled uniformly from the plaintext space. 
        \item $\mathcal{U}_{H}'$ runs $r\leftarrow \mathcal{U}_{H}(c,pk,\tilde{s})$ and outputs $r$ as its own output.\\
    \end{enumerate}
    The simulation converts the difference between the probability that $\mathcal{U}_{H}$ correctly outputs $f(c)$ in the real-world scenario and when invoked by $\mathcal{U}_{H}'$ to a distinguisher distinguishing a legal ciphertext and a "garbage" with known public key components $pk$ and $\tilde{s}$, where the ciphertext and $pk$ are RLWE instances and $\tilde{s}$ is an NTRU instance.
\end{proof}

\begin{definition}\label{def: NTRUinst}\cite{pellet2021hardness}
    A $(\gamma_s,\gamma_t,q)$-NTRU instance is $r \in R_q = Z_q[X]/(X^n+1)$ such that (i) $r = s\cdot t^{-1} \mod{q}$ and (ii) $||s||_2 \leq \frac{\sqrt{q}}{\gamma_s}, ||t||_2 \leq \frac{\sqrt{q}}{\gamma_t} $
\end{definition}
\begin{definition}\label{def: NTRUsearch}\cite{pellet2021hardness}
    The $(\gamma_s,\gamma_t,q)$-search NTRU problem asks, given a $(\gamma_s,\gamma_t,q)$-NTRU instance $r$, to recover $(s,t)\in R^2$ such that (i) $r = s\cdot t^{-1} \mod{q}$ and (ii) $||s||_2 \leq \frac{\sqrt{q}}{\gamma_s}, ||t||_2 \leq \frac{\sqrt{q}}{\gamma_t} $ 
\end{definition}

\subsection{Attacks on the secret key blinding technique}

Given a blinded private key $\tilde{s} = st^{-1}$, along with public parameters $\{d,q,h\}$ for the blinding algorithm, the adversary's goal is to recover either $s$ or $t$. $h$ is not necessarily a public parameter, but we assume it for the simplicity of security discussion. We also denote the set of all elements that $s$ can take as $S_{sk}$ and the security parameter as $\lambda$. Moreover, $\chi_{sk}$ is fixed as the ternary distribution over $R = \mathbb{Z}[X]/X^d+1$, which is a common choice for RLWE-based HE scheme \cite{bossuat2024security}. 



\subsubsection{\textbf{Brute force attacks.}} 
The adversary employing brute force enumeration first selects the smaller set between $S_{sk}$ and $S_{t}$ to iterate \cite{hoffstein1998ntru}. Then the adversary calculates whether $\tilde{s}^{-1}s'$ (or $\tilde{s}t'$) belongs to $S_h$ (or $S_{sk}$) for each $s'$ (or $t'$) encountered during the iteration. Thus, the security $\lambda$ can be represented by the following equation: $\lambda = \min{(\#S_h,\#S_{sk})}$. If the attacker has a sufficiently large space budget, then a meet-in-the-middle approach can be used as an optimization for enumeration. In particular, the adversary decomposes $t'^{-1}$ into $t'^{-1} = t_1+t_2$ for independent traversal, and iteratively computes $t_1\tilde{s}^{-1}$ and $-t_2\tilde{s}^{-1}$. Since $s$ is a polynomial with a small norm, a pair of $(t_1,t_2)$ may be a candidate solution when it satisfise that the results of $t_1\tilde{s}^{-1}$ and $-t_2\tilde{s}^{-1}$ are close. This method works the same for iterating $s'\in S_{sk}$ since $t$ is sparse. Security at this point should be rewritten as $\lambda =  \sqrt{\min{(\#S_h,\#S_{sk})}}$.


\subsubsection{\textbf{Lattice based attacks.}} 

\begin{table}
\begin{tabular*}{\columnwidth}{@{}lccccccccl@{}}\toprule
\multicolumn{3}{c}{Fixed $\log{q'}=27$}  & \phantom{abc}
& \multicolumn{3}{c}{Fixed $d=96$} \\
 \cmidrule{1-3} \cmidrule{5-7} 
$d$  & $||s||_2$ & Time(s) & & $\log{q'}$  & $c_h$ & Time(s)  \\ \midrule
$32$ & $4.62$ & $1.25$             & & $20$  & $0.02$ & $171.88$\\
$48$ & $5.66$ & $6.98$             & & $21$  & $0.02$ & $186.53$\\
$64$ & $6.53$ & $26.27$            & & $22$  & $0.03$ & $184.27$\\
$80$ & $7.30$ & $75.53$            & & $23$  & $0.04$ & $192.08$\\ 
$96$ & $8.00$ & $244.98$           & & $24$  & $0.06$ & $203.35$\\
$112$& $8.64$ & $529.23$           & & $25$  & $0.09$ & $218.02$\\
$128$& $9.24$ & $1824.78 (\times)$ & & $26$  & $0.12$ & $253.43$\\
\bottomrule
\\
\end{tabular*}
\caption{Lattice Reduction for $L(\tilde{s},\alpha)$ with $h = 5$ and $\log{q} = 33$}
\label{tb: LR_Q33R27H5}
\end{table}

\begin{table}
\begin{tabular*}{\columnwidth}{@{}lccccccccl@{}}\toprule
\multicolumn{3}{c}{Fixed $\log{q'}=27$}  & \phantom{abc}
& \multicolumn{3}{c}{Fixed $d=96$} \\
 \cmidrule{1-3} \cmidrule{5-7} 
$d$  & $||s||_2$ & Time(s) & & $\log{q'}$  & $c_h$ & Time(s)   \\ \midrule
$32$ & $4.62$ & $1.32$             & & $20$  & $0.02$ & $149.01$\\ 
$48$ & $5.66$ & $6.95$             & & $21$  & $0.03$ & $187.77$\\
$64$ & $6.53$ & $44.05$            & & $22$  & $0.04$ & $206.77$\\
$80$ & $7.30$ & $75.62$            & & $23$  & $0.05$ & $217.06$\\ 
$96$ & $8.00$ & $384.09(\times)$   & & $24$  & $0.07$ & $226.95$\\
$112$& $8.64$ & $991.29(\times)$   & & $25$  & $0.10$ & $232.05$\\
$128$& $9.24$ & $2074.43(\times)$  & & $26$  & $0.14$ & $276.07$\\
\bottomrule
\\
\end{tabular*}
\caption{Lattice Reduction for $L(\tilde{s},\alpha)$ with configured by $h = 10$ and $\log{q} = 33$}
\label{tb: LR_Q33R27H10}
\end{table}

Lattice reduction is a common technique for solving NTRU-search problems. Given the blinded secret key $\tilde{s} = st$, one can create a lattice containing some short vectors related to $s$ and $t$, which can be found by reduction algorithms such as BKZ\cite{schnorr1994lattice,chen2011bkz}. To construct such a lattice, we start by creating a $2d\times 2d$ matrix $L(\tilde{s}) = \begin{pmatrix} \mathbf{I} & \mathbf{\tilde{S}} \\ \mathbf{0} & q\mathbf{I} \end{pmatrix} $, where the 
$i$-th row of $\mathbf{\tilde{S}}$ are generated by the nega-cyclic multiplication between $\tilde{s}$ and $x^i$: 
\begin{equation}
\mathbf{\tilde{S}} = \begin{pmatrix}
\tilde{s}_0 & \tilde{s}_1 & \dots & \tilde{s}_{d-1} \\
\tilde{s}_{d-1} & \tilde{s}_0 & \dots & \tilde{s}_{d-2} \\
\vdots & \vdots & \ddots & \vdots \\
\tilde{s}_1 & \tilde{s}_2 & \dots & \tilde{s}_0 
\end{pmatrix}. 
\end{equation}
The vector $\tau = (t,s)$ is in the lattice generated by the rows of $L$ since we have $(t,s) = tL \mod{q}$. To make sure the lattice reduction algorithms have a higher probability of locating $\tau$, we modify $L(\tilde{s})$ into:
\begin{equation}
L(\tilde{s},\alpha) = \begin{pmatrix}
\alpha\mathbf{I} & \mathbf{\tilde{S}} \\ \mathbf{0} & q\mathbf{I} 
\end{pmatrix}. 
\end{equation}
Here, $\tau$ is redefined as $\tau = (\alpha t,s)$, and $\alpha=|s|_2/|t|_2$ so that the ratio of $|\tau|_2$ to the expected lower bound of the shortest vector in $L(\tilde{s},\alpha)$, denoted as $c$, is set to minimum: $c = \sqrt{\frac{2\pi e ||s||_2\cdot ||t||_2}{dq}}$. 

We perform lattice reduction on low dimension $L(\tilde{s},\alpha)$ using the well-known BKZ algorithm implemented in the Sagemath library\cite{sage} with fixed blocksize $\beta = 20$ (see Tables \ref{tb: LR_Q33R27H5} and \ref{tb: LR_Q33R27H10}). During the experiment, $t$ generated with $\Pi_{SkBdKeyGen}(R_{q},h)$ is bounded by an extra coefficient modulus $q'\leq q$ to see how $t$'s length affects the reduction performance. It can be observed that for fixed $d$ and $q$, increasing $t$'s length by raising $q'$ or $h$ makes it harder to find the target vector. This corresponds to a $c$ closer to $1$, which implies that $L(\tilde{s},\alpha)$ behaves more like a random lattice. On the other hand, $d$ also significantly impacts the reduction as the time seems to be exponentially high corresponding to $d$.


As $d$ inherited from the RLWE-based HE is usually set to a value no smaller than $2^{10}$, it is difficult to conduct lattice reduction experiments on such a high-dimensional setting in the real world. Thus, the best we can do is to estimate the complexity of reduction. To do so, we first compute the Hermite factor of $L(\alpha t,s)$: $\delta = ||\mathbf{b}_i||_2 / \text{vol}(L(\alpha t,s))$. BKZ with a larger block size is more likely to find the shortest vector in a lattice with a smaller Hermite factor, where the relationship between the Hermite factor and the block size can be described with the following equations\cite{chen2011bkz,dachman2020lwe}:
\begin{equation}\label{eq: delta-beta}
    \delta^{\text{dim}(L)} = \left( (\pi \beta)^{\beta/1} \cdot \frac{\beta}{2\pi e} \right)^{1/(2\beta-2)}
\end{equation}
In our case, if the BKZ finds a vector in $L(\alpha t,s)$ no longer than the expected shortest vector given by the Gaussian heuristic, then it reaches a Hermite factor of $\delta = \sqrt{\frac{d}{\pi e}}^{1/(2d)}$ \cite{gama2008predicting}. However, one can observe that $\delta$ is no larger than $0.0005$ when $d>2^{10}$, and the estimated blocksize for such $\delta$ given by Equation \ref{eq: delta-beta} is larger than the dimension of the lattice. This implies that BKZ may have difficulty finding a short vector in a reasonable time since the logarithm of the number of elementary operations in BKZ is roughly polynomial to the block size.

\subsubsection{\textbf{Zero-forced attacks.}} Although the general lattice-based attack seems unable to pose a significant threat to the protocol's security, we observe that it can be enhanced using the so-called zero-forced attack \cite{may1999cryptanalysis,silverman4dimension}, considering the high sparsity of $t$. The idea is to guess $r$ specific coefficient indices with zero values in $t$ ( or any of its nega-cyclic shifting variants) and to construct a lattice with reduced dimension $2(d-r)$. If the guess is correct, the dimension-reduced lattice will contain the expected short vector. In our scenario, this lattice can be constructed using the following steps:

\noindent\textbf{Step 1.}\quad Randomly select a set of specific coefficient indices $J$ as our guess:
\begin{equation}
    J = \{ j_1,j_2,\dots,j_r | 0\leq j_1 < \dots <j_r <N \} 
\end{equation}

\noindent\textbf{Step 2.}\quad Construct a variant of $\mathbf{\tilde{S}}$ with dimension $2\cdot (d-r)$, deonted as $\mathbf{\tilde{S}}^{ZF}$. The $i$-th row of this matrix is composed of the coefficients of $\tilde{s}$, which contribute to the coefficient $\tilde{s}_i$. There are exactly $d-r$ coefficients satisfying this condition since we assume that the coefficients of $t$ at positions $J$ are all zero. Formally, the $i$-th row of $\mathbf{\tilde{S}}^{ZF}$ can be expressed using the following equation:
\begin{equation}
    (\mathbf{\tilde{S}}^{ZF})_i = \begin{bmatrix}
       \cdots & (\tilde{s}\cdot x^{i})_{j} & \cdots  
    \end{bmatrix}, 0\leq j < d, j\not\in  J
\end{equation}
Note that we choose to guess the positions of zero coefficients in $t$ instead of in $s$ because in our scheme $t$ should be much sparser than $s$ which may be sampled from some discrete Gaussian distribution or ternary distribution.

\noindent\textbf{Step 3.}\quad Construct the dimension-reduced lattice with the following form:
\begin{equation}
    L(\alpha,\tilde{s})^{ZF} = \begin{bmatrix}
        \alpha I & \mathbf{\tilde{S}}^{ZF} \\
        \mathbf{0} & qI
    \end{bmatrix}
\end{equation}
Note that $L(\alpha,\tilde{s})^{ZF}$ is now a lattice with size $2\cdot(d-r) \times 2\cdot(d-r)$

The probability for us to guess a correct set of $J$ for $t$ or any of its shifted versions is computed with the following equation proposed by Silverman \cite{silverman4dimension}:
\begin{equation}\label{eq: prob}
    \text{Prob} \begin{pmatrix}
        \exists 0\leq k <N: \\ t_{j_1+k} = \dots = t_{j_r+k} = 0
    \end{pmatrix}  \approx 1-(1-\prod_{i=0}^{m-1}(1-\frac{r}{d-1}))^d . 
\end{equation}
Taking the result of Equation (\ref{eq: prob}) as $p$, the expected time for a zero-forced attack with specific $r$ to find $\tau$ or any of its shifting versions in some $L(\alpha,\tilde{s})^{ZF}$ is computed by $1/p\cdot \text{oprs}^{ZF}$, where $\text{oprs}^{ZF}$ denotes the number of operations required for reducing one $L(\alpha,\tilde{s})^{ZF}$ instance. Table \ref{fig: zeroforced_h} demonstrates the probability $p$ and the zero-forced attack complexity as $r$ increases with different choices of $h$. The complexity of the attack corresponds to the number of elementary operations in one BKZ iteration\cite{hoffstein2017choosing,chen2011bkz}: $\log{(\text{oprs})} = poly(\beta)+\log{dim}+7$, where $poly(\beta) = 0.00405892\beta^2 - 0.337913\beta+34.9018$.

\begin{figure}
\centerline{\includegraphics[scale=0.50]{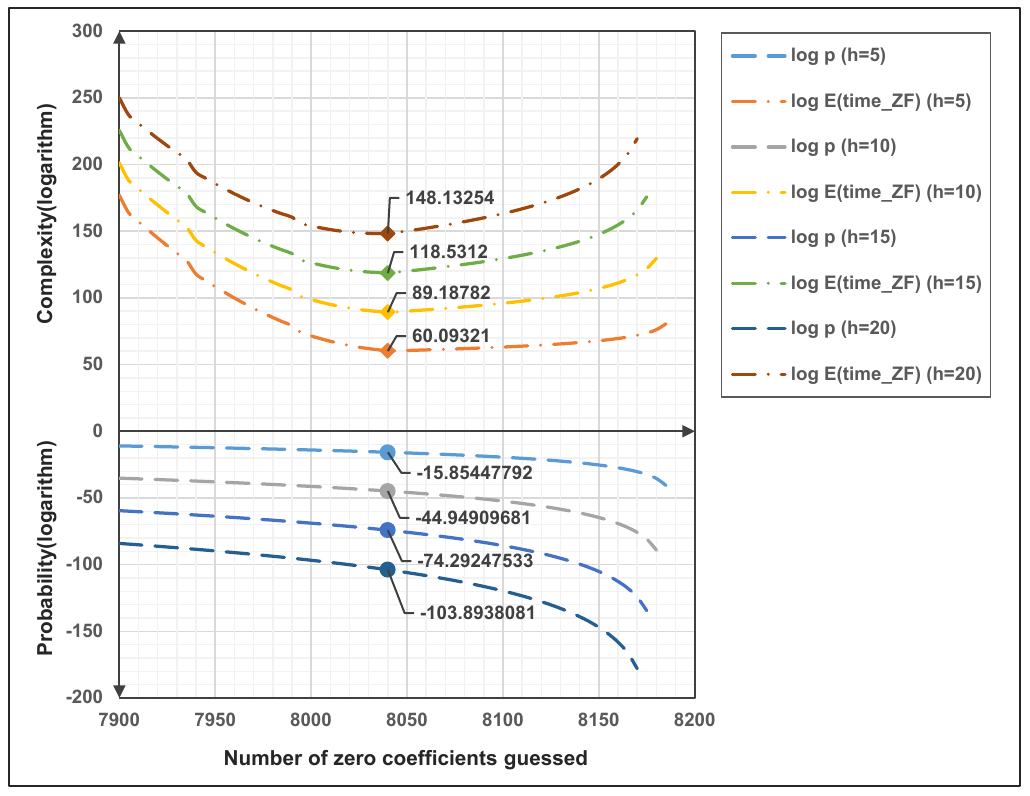}} 
\caption{Zero-forced attack complexity}
\label{fig: zeroforced_h}
\end{figure}

It can be observed that when the Hamming weight of $t$ is only 5, the complexity of the zero-forced attack drops to $2^{60}$ operations, posing a serious threat to the protocol's security. To ensure 128-bit security, a Hamming weight greater than 14 is required. However, higher density for $t$ implies increased complexity for the local decryption algorithm.

To counter this attack, we propose modifying the generation method of $t$: let $t$ be a composite polynomial formed by variables from the original sampling method. Specifically, for $h = h_1 + h_2$, we can set $\leftarrow \Pi_{\text{SkBdKeyGen}}(R_q, h_1) \cdot \Pi_{\text{SkBdKeyGen}}(R_{q_2}, h_2) \mod{q}$. This way, the Hamming weight of $t$ is at least $h_1 \cdot h_2 - \min(h_1, h_2)$, and this lower bound is only met when the non-zero positions of the two composite parts align exactly in arithmetic sequences with the same step size. One can adjust the number of components composing $t$ and each of their hamming weight to achieve the trade-off between complexity and density. Overall, the key idea is to enhance the density of $t$ without intensively increasing the complexity of polynomial multiplication. Note that the increased non-zero positions of the modified $t$ are not independent. Nonetheless, we believe this modification is sufficient to withstand zero-forced attacks.



\subsection{Security Parameterization}


Based on the attacks discussed above, we can parameterize the security of the outsourced decryption. We consider three security parameter cases: $\lambda = 128, 192, 256$. For each $\lambda$, we discuss the required values of $h$ and $q$ when $d = 13, 14, 15$, and $16$. During the parameterization, $h$ is selected such that a zero-forced attack requires a complexity no smaller than $2^{\lambda}$ elementary operations. For a given $h$, we fix the sampling method of $t$ as $t \leftarrow \Pi_{\text{SkBdKeyGen}}(R_q, h_1) \cdot (\Pi_{\text{SkBdKeyGen}}(R_{q_2}, h_2)) \mod{q}$, where $h_1\cdot h_2 - \text{min}(h_1,h_2) \geq h$. The sampling space determined by $q$ and $h$ should be large enough to guard against brute-force attacks. For adversaries enumerating the $t$ using the sampling method above, $q$ and $h$ should satisfy the following equation:
\begin{equation}
 \sqrt{C_d^{h_1} (q-1)^{h_1} \cdot C_d^{h_2} (q_2-1)^{h_2}} \geq 2^{\lambda}.  
\end{equation}
By further fixing $h_1= 6$ and $q_2 = 2$, we generate the sample parameter sets satisfying three security categories respectively, as shown in Table \ref{tb:params}.

\begin{table}
\begin{tabular*}{\columnwidth}{@{}ccccccccccccc@{}} \toprule
& & \multicolumn{2}{c}{$\lambda=128$}  & \phantom{abc} & \multicolumn{2}{c}{$\lambda=192$}  & \phantom{abc} & \multicolumn{2}{c}{$\lambda=256$} \\
\cmidrule{3-4} \cmidrule{6-7} \cmidrule{9-10}
$d$ & & $h$ & $\log{q}$ & & $h$ & $\log{q}$ & & $h$ & $\log{q}$ \\ \midrule
$2^{13}$ & & $17$ & $23$ & & $28$ & $19$ & & $39$ & $16$   \\
$2^{14}$ & & $15$ & $22$ & & $25$ & $19$ & & $34$ & $15$  \\
$2^{15}$ & & $13$ & $22$ & & $22$ & $18$ & & $30$ & $15$  \\
$2^{16}$ & & $12$ & $21$ & & $19$ & $18$  & & $26$ & $13$  \\
\bottomrule
\end{tabular*}    
\\
\caption{parameterization for the outsourced decryption protocol}\label{tb:params}
\end{table}

\section{Implementation}\label{sec:exp}

After discussing the security of the outsourced decryption protocol, we move on to the implementation and performance evaluation. 

We implement all the proposed algorithms with the Lattigo library V4.1.0 \cite{lattigo}, where we use the CKKS scheme it supports as the underlying RLWE-based HE scheme. Our experiments are conducted on a machine equipped with an i7-13700K (3.40 GHz) processor and 64GB of memory.

To evaluate the protocol's performance on practical RLWE-based HE parameter sets, we compare the local decryption time of the original decryption and our outsourced method under parameter sets with $d={13,14,15,16}$ (see Figure \ref{fig: ckks1_3} and \ref{fig: ckks1_4} for comparison results). One may refer to the FHE parameterization guidance proposed by Bossuat et al.\cite{bossuat2024security} for the corresponding security level under the RLWE assumption. Note that for any fixed \( d \), the security of the RLWE assumption increases as the modulus \( q \) decreases. The minimum required \( q \) in Table \ref{tb:params} is sufficient to exceed 256-bit security under the RLWE assumption. Thus, in terms of parameter selection, the outsourced decryption and the underlying HE scheme exhibit no explicit conflict.



\begin{figure}
\centerline{\includegraphics[scale=0.47]{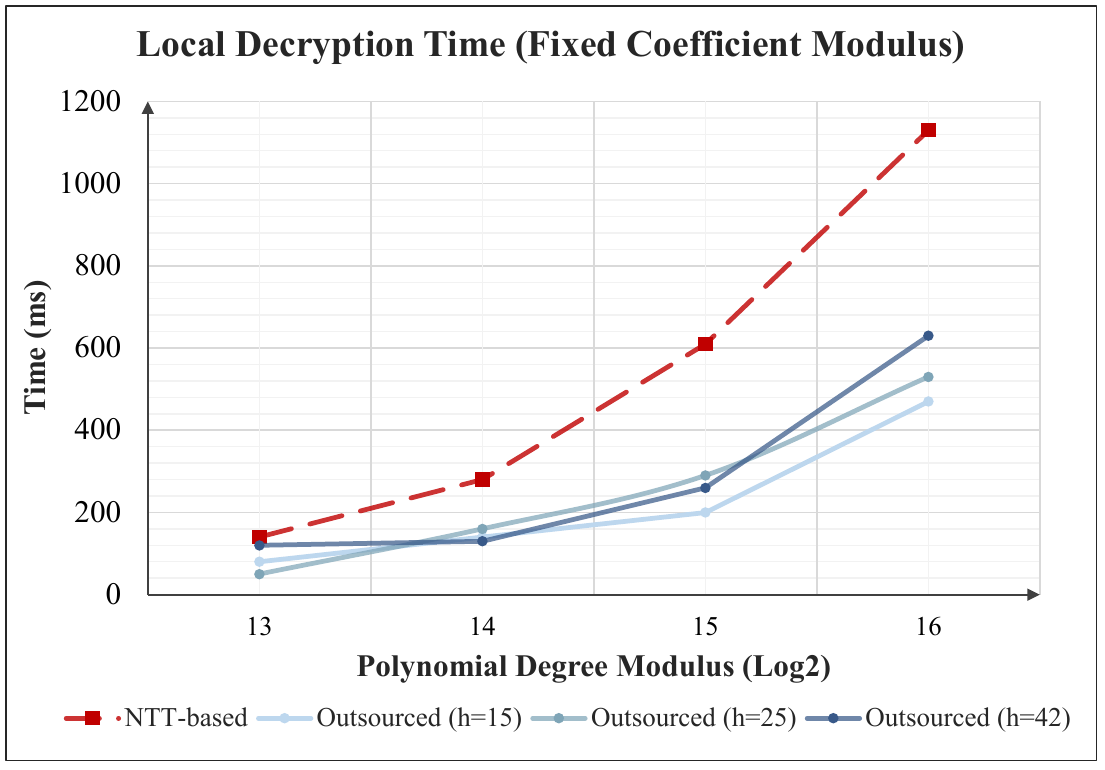}} 
\caption{Local decryption time with $\log{q}\leq 64$}
\label{fig: ckks1_3}
\end{figure}

\begin{figure}
\centerline{\includegraphics[scale=0.47]{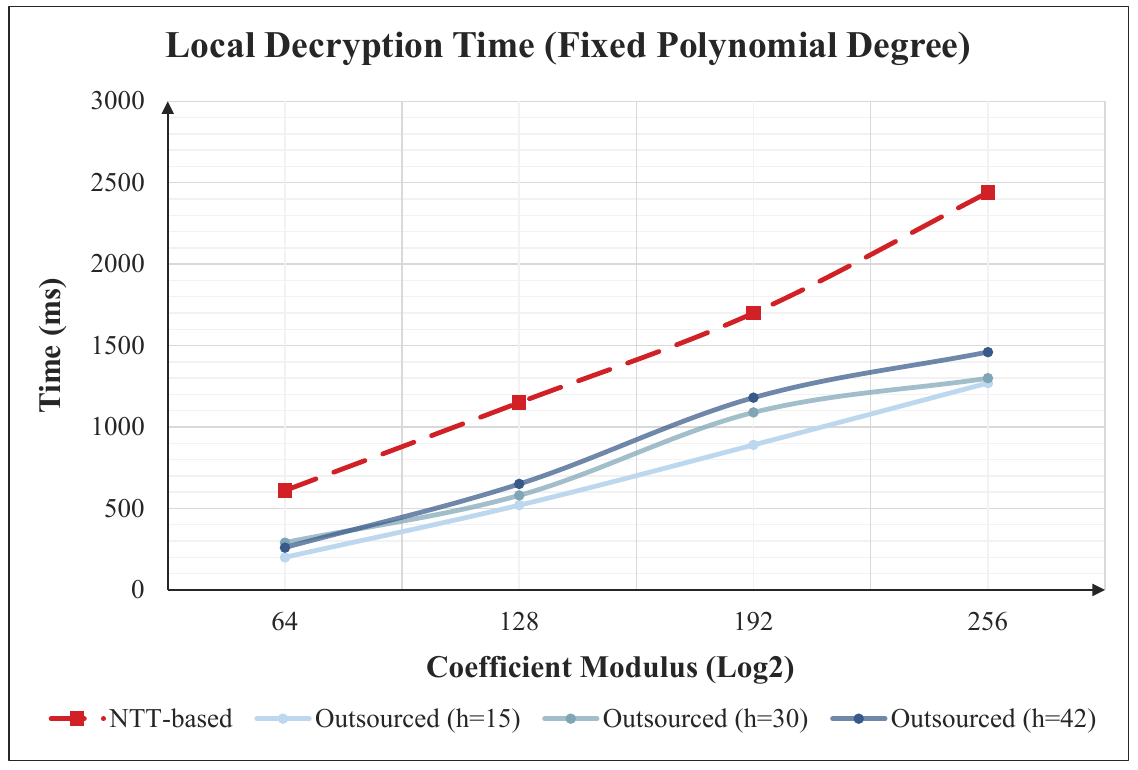}} 
\caption{Local decryption time with fixed $d=2^{15}$}
\label{fig: ckks1_4}
\end{figure}

\subsection{Efficiency}



Figure \ref{fig: ckks1_3} illustrates the performance of local decryption across various polynomial degree moduli, where we test the cumulative time for executing local decryption 1000 times. For different $d$, we select coefficient moduli $q$ differently to satisfy the NTT requirement $q \mod {2d} \equiv 1$. Since Lattigo uses at least one 64-bit integer to store a polynomial coefficient modulo $q$, we control $\log{q}<64$ to eliminate the complexity difference brought by $q$'s length. In fact, for $q$ satisfying the residue number system: $q = \prod_{i=0}^{L-1}q_i$, Lattigo employs 64-bit integers to represent each $q_i$.

It can be observed that the differences between the curves for various $h$ are minimal, as in the sampling method of $t$, with $h_1$ fixed at $6$, $h_2$ only needs to be within the range $[3,8]$ to meet all three security categories. On the other hand, the time taken by the original decryption algorithm based on NTT increases with $d$, while the gap between this and our protocol's local decryption time also grows, corresponding to the factor $1/\log{d}$ in the efficiency ratio. Overall, we achieve a speed-up of $43\%\sim 67\%$ under $\lambda=128$, and $14\% \sim 67\%$ under $\lambda=256$.

We further conducted 1000 iterations of local decryption with a fixed $d$ and varying coefficient moduli (see Figure \ref{fig: ckks1_4}). We can observe that the efficiency ratio between our protocol and the original decryption algorithm remains fairly consistent across different coefficient moduli lengths, aligning with our theoretical predictions. Note that homomorphic evaluation is typically performed on larger $q$, but as long as $q$ satisfies the residue number system, decryption can be performed in a low-modulus ring while maintaining a considerable level of precision.

\subsection{Space Reduction}

We observe that the proposed protocol helps save the memory needed for decryption. For the original decryption, its minimum space consumption commonly consists of three components: secret key $s \in R_q$, ciphertext $(u,v)\in R_q^2$, and parameters for performing $INTT$. $(u,v)$ takes up $2 \ell d$ bits, where $\ell$ is the number of bits for storing one coefficient modulo $q$. Although $s$ is usually sampled from some distribution with very small deviation, it is the point-value form $\hat{s} \leftarrow INTT(s)$ instead of $s$ itself that should be stored. $\hat{s}$ requires $\ell d$ bits since the components of $\hat{s}$ is not necessarily small. The parameters for $INTT$ are actually powers of some root of unity in $R_q$ and also need $\ell d$ bits of memory. Such space composition can be relatively large due to large $d$ and $q$. Table \ref{fig: ckks1_1} demonstrates the space consumption for original decryption implemented by the Lattigo library with $d=2^{15}$.

On the other hand, $\Pi_{\text{LocDec}}$ requires no invocation of $\hat{s}$ and NTT-related parameters, but only the sparse unblinding factor $t$ and the blindly decrypted ciphertext $(\tilde{u},v)$. This saves up to nearly $50\%$ of space consumption since $t$ requires only $2 \ell (h_1+h_2)$ bits for storing its non-zero coefficients and the corresponding indices, where $h < h_1*h_2 - \text{min}(h_1,h_2)$ and $h$ is much smaller than $d$. Table \ref{fig: ckks1_4} demonstrates the space consumption for the local decryption with $d=2^{15}$ and various $q$. 


\begin{figure}
\centerline{\includegraphics[scale=0.60]{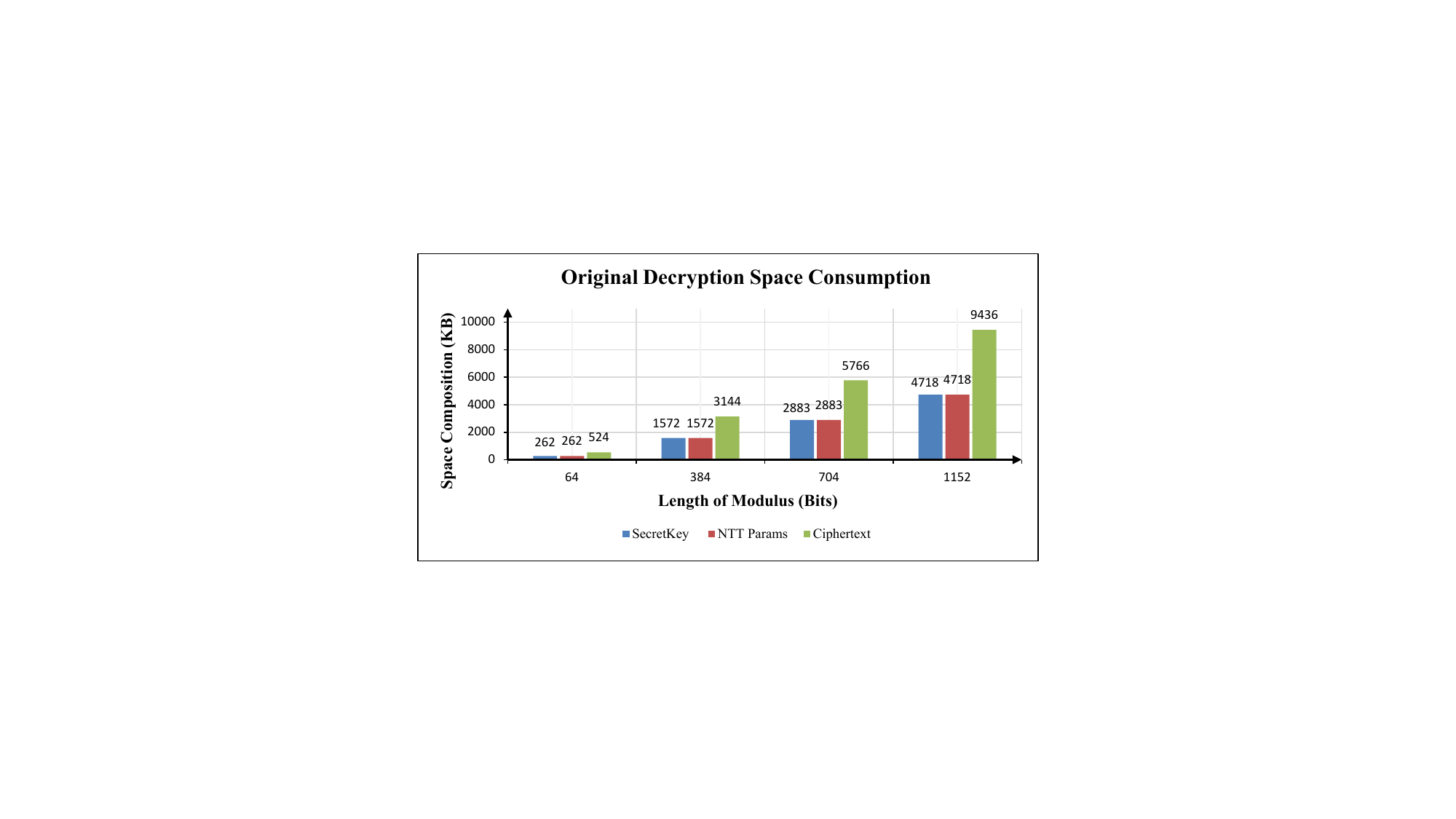}} 
\caption{Original decryption space consumption with $d=2^{15}$}
\label{fig: ckks1_1}
\end{figure}

\begin{figure}
\centerline{\includegraphics[scale=0.60]{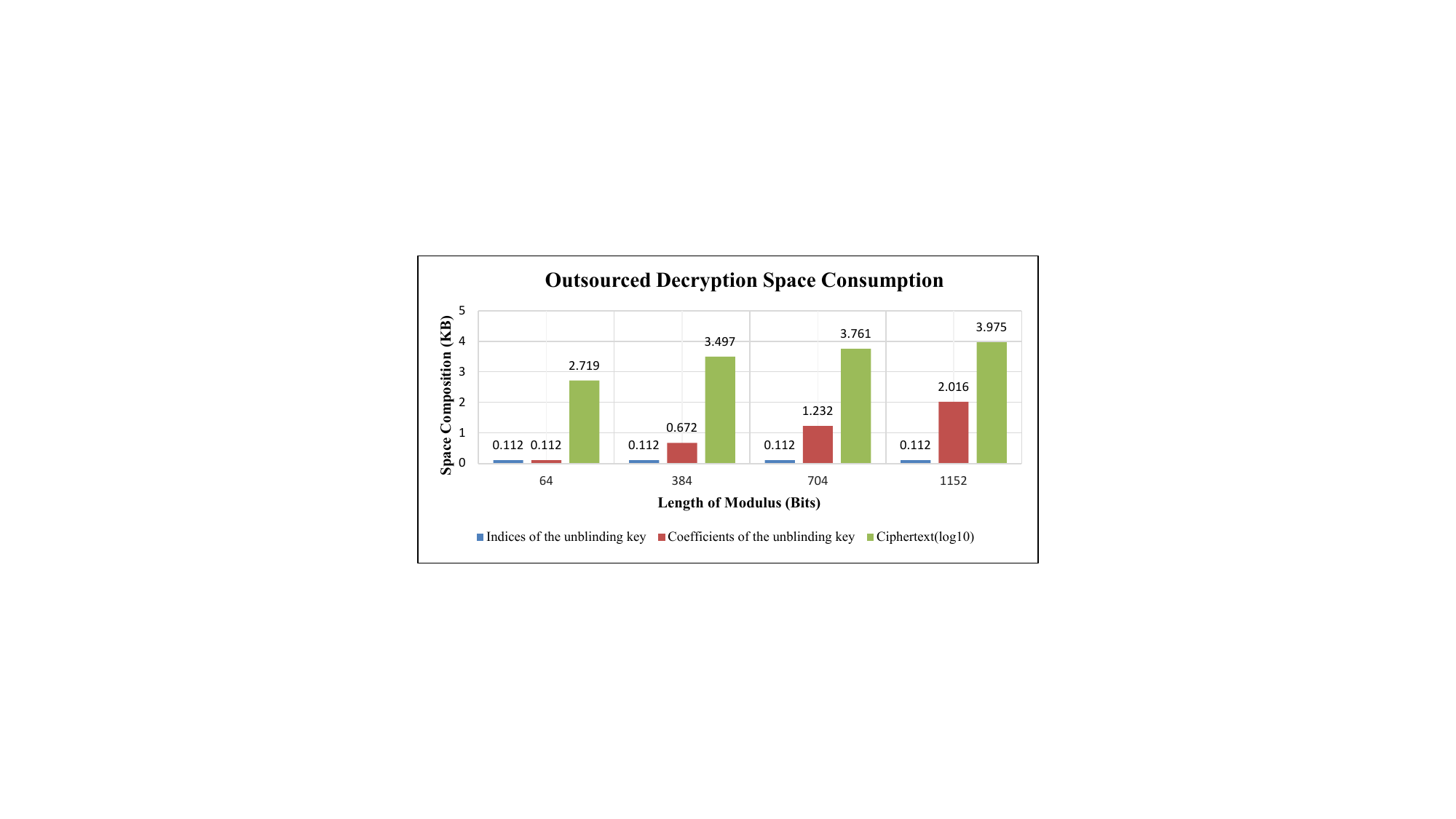}} 
\caption{Outsourced decryption space consumption with $d=2^{15}$}
\label{fig: ckks1_2}
\end{figure}

\section{Conclusion \& Further Discussion}\label{sec: conclusion}
In this paper, we have presented an outsourced decryption protocol for RLWE-based homomorphic encryption schemes to alleviate clients' local decryption burden in a cloud computing setting. We have demonstrated that the protocol preserves certain security under typical attacks and exhibits a significant speed-up and space reduction on client-side decryption. Nevertheless, there are several topics for further discussion. 

The security of this protocol could be further investigated. For example, subring and subfield attacks \cite{albrecht2016subfield,cheon2016algorithm,duong2017choosing,kirchner2017revisiting} compromise NTRU in polynomial time when the coefficient modulus is exponentially high relative to the polynomial degree, but with the secret vector occupying a much smaller space. Thus it is uncertain how they may affect the security of the proposed secret key blinding technique. Moreover, the proposed protocol is secure only if the RLWE and NTRU assumptions hold when both instances are published under the same secret. Therefore exploring the actual effect of such behavior is also required in future work.

Moreover, a verification method can be provided for the protocol to adapt it to more variable scenarios. One possible probabilistic verification method is to use polynomial evaluation. Client can prepare some random points, and verify correctness by comparing the evaluation results returned by the cloud with those of the original decryption algorithm at these points. Further discussion is needed to determine the feasibility of this approach.



\section*{Acknowledgment}
This work is supported by the Key Research and Development Program of Shandong Province, China (Grant No. 2022CXGC020102), the Science and Technology Small and Medium Enterprises (SMEs) Innovation Capacity Improvement Project of Shandong Province \& Jinan City, China (Grant No. 2022TSGC2048), Haiyou Famous Experts - Industry Leading Talent Innovation Team Project, and the National Natural Science Foundation of China under Grant(NO.62072247).

\bibliographystyle{IEEEtran}
\bibliography{IEEEabrv,mylib}

\clearpage

\end{document}